\documentclass{llncs}

\usepackage{balance}
\usepackage{mathpartir}
\usepackage{amsmath}
\usepackage{hyperref}
\usepackage{amssymb}
\usepackage{listings}
\usepackage{atbegshi}
\usepackage{ifthen}
\usepackage[dvipsnames]{xcolor}
\usepackage{booktabs}
\usepackage{url}
\usepackage{xspace}
\usepackage{tikz}
\usepackage{tikzscale}
\usepackage{paralist}
\usepackage{subcaption}
\usepackage{verbatim}
\usepackage[numbers,sort]{natbib}
\usepackage{wrapfig}
\usepackage{booktabs}
\usepackage{multirow}
\usepackage{graphicx}
\usetikzlibrary{calc}
\usetikzlibrary{shapes,arrows,positioning}
\usepackage[vlined, ruled, commentsnumbered, linesnumbered]{algorithm2e}
\SetKw{Continue}{continue}
\newcommand{\comments}[1]{}
\newcommand{\code}[1]{\lstinline|#1|}

\newcommand{\false}{\textit{false}}

\newcommand{\translate}[2]{[#1]_{#2}}
\newcommand{\negencode}[2]{\textsc{Encode}(#1, #2)}
\newcommand{\posencode}[3]{\textsc{Unroll}(#1, #2, #3)}
\newcommand{\step}[3]{\textsc{Step}(#1, #2, #3)}
\newcommand{\unsat}{\textit{unsat}\xspace}

\newcommand{\para}[1]{\vspace{4pt}\noindent{\bf #1.}\hspace{4pt}}
\newcommand{\pred}[1]{{\small \texttt{#1}}}

\newcommand{\tcprog}{P_{tc}}
\newcommand{\tcprop}{\varphi_{tc}}

\newcommand{\sem}[1]{[\![#1]\!]}
\newcommand{\idb}[1]{\textit{idb}({#1})}
\newcommand{\edb}[1]{\textit{edb}({#1})}
\newcommand{\getvars}[1]{\textit{vars}(#1)}

\newcommand{\vars}{\textit{Vars}}
\newcommand{\vals}{\textit{Vals}}
\newcommand{\preds}{\textit{Preds}}

\newcommand{\synthbox}{\relax\ifmmode {\cal S}_\textit{SemiPos} \else ${\cal S}_\textit{SemiPos}$\fi \xspace}
\newcommand{\synth}{\relax\ifmmode {\cal S}_\textit{Strat} \else ${\cal S}_\textit{Strat}$\fi \xspace}

\newcommand{\badinputs}{{\cal F}}

\newcommand{\model}{{\cal M}}

\newcommand{\encodepred}[2]{\textsc{EncodePred}(#1, #2)}

\newcommand{\reqs}{\relax\ifmmode \varphi_R \else $\varphi_R$\fi \xspace}
\newcommand{\net}{\relax\ifmmode \varphi_N \else $\varphi_N$\fi \xspace}
\newcommand{\conf}{\relax\ifmmode \varphi_C \else $\varphi_C$\fi \xspace}

\definecolor{darkred}{rgb}{0.75,0,0}
\definecolor{darkgreen}{rgb}{0.0, 0.5, 0}
\definecolor{lightgreen}{rgb}{0.8,1,0.8}
\definecolor{lightred}{rgb}{1,0.8,0.8}
\definecolor{lightblue}{rgb}{0.8,0.8,1}

\tikzstyle{highlighter} = [
  yellow,
  line width = \baselineskip,
]

\newcounter{highlight}[page]

\AtBeginShipout{\AtBeginShipoutUpperLeft{\ifthenelse{\value{highlight} > 0}{\tikz[remember picture, overlay]{\foreach \stroke in {1,...,\arabic{highlight}} \draw[highlighter] (begin highlight \stroke) -- (end highlight \stroke);}}{}}}
 
\pgfdeclareimage[height=16pt, width=16pt]{gears}{figures/gears.png}

\definecolor{mygray}{gray}{0.5}
\begin{document}

\setlength{\abovedisplayskip}{4pt}
\setlength{\belowdisplayskip}{4pt}
\setlength{\belowcaptionskip}{-10pt}
\newcommand{\petar}[1]{{\color{darkred}\bf $\blacktriangleright$#1$\blacktriangleleft$}}
\newcommand{\ah}[1]{{\color{red}\bf $\blacktriangleright$#1$\blacktriangleleft$}}
\newcommand{\tool}{\texttt{SyNET}}
\newcommand*{\DashedArrow}[1][]{\mathbin{\tikz [baseline=-0.25ex,-latex, dashed,#1] \draw [#1] (0pt,0.5ex) -- (1.3em,0.5ex);}}
\newif\ifarxiv
\arxivtrue

\title{Network-wide Configuration Synthesis}
\author{Ahmed El-Hassany\and Petar Tsankov\and  Laurent Vanbever\and Martin Vechev}
\institute{ETH Z\"urich}

\maketitle

\begin{abstract}
Computer networks are hard to manage. Given a set of high-level requirements
(e.g., reachability, security), operators have to manually figure out
the individual configuration of potentially hundreds of devices running complex
distributed protocols so that they, collectively, compute a compatible
forwarding state. Not surprisingly, operators often make mistakes which lead to
downtimes.

To address this problem, we present a novel synthesis approach that automatically computes correct network configurations that comply with the operator's requirements.
We capture the behavior of existing routers along with the distributed protocols they run in stratified Datalog. Our key insight is to reduce the problem of finding correct input configurations to the task of synthesizing inputs for a stratified Datalog program.

To solve this synthesis task, we introduce a new algorithm that synthesizes inputs for stratified Datalog programs. This algorithm is applicable beyond the domain of networks.

We leverage our synthesis algorithm to construct the first network-wide configuration synthesis system, called \tool, that support multiple interacting routing protocols (OSPF and BGP) and static routes. We show that our system is practical and can infer correct input configurations, in a reasonable amount time, for networks of realistic size ($>50$ routers) that forward packets for multiple traffic classes.

\end{abstract}
 \section{Introduction}
Despite being mission-critical for most organizations, managing a network is
surprisingly hard and brittle. 

A key reason is that network operators have to manually come up with a configuration, which ensures that the underlying distributed protocols
compute a forwarding state that satisfies the operator's requirements.

Doing so requires operators to precisely understand: \emph{(i)} the behavior of each distributed protocol;
\emph{(ii)} how the protocols interact with each other;
and \emph{(iii)} how each
parameter in the configuration affects the distributed computation.

Because of this complexity, operators often make mistakes that can lead to
severe network downtimes. As an illustration, Facebook (and
Instagram) recently suffered from widespread issues for about an hour due to a
misconfiguration~\cite{facebook_down}. In fact, studies show that most
network downtimes are caused by humans, not equipment failures~\cite{juniper_downtime}. Such
misconfigurations can have Internet-wide effects~\cite{bgpmon}.

To prevent misconfigurations, researchers have developed tools that check
if a given configuration is correct~\cite{Fogel:2015:GAN:2789770.2789803, rcc, margrave, fireman}. While useful, these works still require network operators to
produce the configurations in the first place. Template-based
approaches~\cite{bgp_template, chen2009pacman, enck2009configuration, 1248660} along with vendor-agnostic abstractions~\cite{rpsl,yang,netconf} have been
proposed to reduce the configuration burden. However,
they
still require operators
to understand precisely the details of each protocol. Recently, Software-Defined
Networks (SDNs) have emerged as another paradigm to manage networks by
\emph{programming} them from a central controller.
Deploying SDN is, however,
a major hurdle as it requires new network devices \emph{and} management tools.
Further, designing correct, robust and yet, scalable, SDN controllers is challenging~\cite{SDNRacer, nice, Scott:2014:STS, vericon}.
Because of this, only a handful of networks are using SDN in production.
As a result, configuring individual devices is by far the most widespread (and default) way to manage networks.

\para{Problem Statement: Network-Wide Configuration Synthesis} Ideally, from a network operator perspective, one would like to solve what we refer to as the \emph{Network-Wide Configuration Synthesis} problem: \emph{Given a network specification $\mathcal{N}$, which defines the behavior of all routing protocols run by the routers, and a set $\mathcal{R}$ of requirements on the network-wide forwarding state, discover a configuration $\mathcal{C}$ such that the routers converge to a forwarding state compatible with~$\mathcal{R}$.}
That is, the operator simply provides the high-level requirements $\mathcal{R}$, and the configuration $\mathcal{C}$ is obtained automatically.

\para{Distributed vs. Static routing}
Relying as much as possible on distributed protocols to compute the forwarding state is critical to ensure network reliability and scalability. 
A simpler problem would be to 
statically configure the forwarding entries of each router via static routes (e.g. see~\cite{Subramanian:2017:GSF:3009837.3009845,Kang:2013:OOB:2535372.2535373}).
Relying solely on static routes is, however, undesirable for two reasons. First, they prevent routers from reacting locally upon failure.
Second, they can be costly to update as routers often have a large number of static entries.

\para{Key Challenges}
Coming up with a solution to the network-wide synthesis problem is challenging for at least three reasons:
\emph{(i)~Diversity}:
protocols have different expressiveness in terms of the forwarding entries they compute. For instance, the Open Shortest Path First protocol (OSPF) can only direct traffic along shortest-paths, while the Border Gateway Protocol (BGP) can direct traffic along non-shortest paths.
Conversely, BGP cannot forward traffic along multiple paths by default\footnote{While vendor-specific workarounds to make BGP multipath exist, these break the congruency between the control and data plane and could lead to correctness issues.}, while OSPF supports multi-path routing and is thus better suited for load-balancing traffic, a feature heavily used in practice.
\emph{(ii)~Dependence:} distinct protocols often depend on one another, making it
challenging to ensure that they collectively compute a compatible forwarding state.
For instance,
BGP depends on the network-wide intra-domain configuration;
and \emph{(iii)~Feasibility}: the search space of configurations is massive and it is thus difficult to find one that leads to a forwarding state satisfying the requirements. 

\para{This Work}
In this paper, we provide the first solution to the network-wide synthesis problem. Our approach is based on two steps. First, we use stratified Datalog to capture the behavior of the network, i.e. the distributed protocols ran by the routers together with any protocol dependencies.
Datalog is indeed particularly well-suited for describing these protocols in a clear and declarative way. Here, the fixed point of a Datalog program represents the stable forwarding state of the network.
Second, and a key insight of our work: we pose the network-wide synthesis problem as an instance of finding an input for a stratified Datalog program where the program's fixed point satisfies a given property.
That is, the network operator simply provides the high-level requirements $\mathcal{R}$ on the forwarding state (i.e., which is the same as requiring the Datalog program' fixed point to satisfy $\mathcal{R}$), and our synthesizer automatically finds an input $\mathcal{C}$ to the Datalog program (i.e., which identifies the wanted network-wide configuration).
We remark that our Datalog input synthesis algorithm is a general, independent contribution, and is applicable beyond networks.

\para{Main Contributions} To summarize, our main contributions are:
\begin{itemize}
	\item A formulation of the network-wide synthesis problem in terms of input synthesis for stratified Datalog (Section~\ref{sec:overview}).
	\item The first input synthesis algorithm for stratified Datalog. This algorithm is of broader interest and is applicable beyond networks (Section~\ref{sec:datalog-synthesis}).
	\item An instantiation and an end-to-end implementation of our input synthesis algorithm to the network-wide synthesis problem, along with network-specific optimizations, in a system called \tool.
    \item An evaluation of \tool~on networks with multiple interacting widely-used protocols.
	In addition, we test the correctness of the generated configurations on an emulated network environment.
    Our results show that \tool~can automatically synthesize input configurations for networks of realistic size ($>50$ routers) carrying multiple traffic classes (Section~\ref{sec:evaluation}).
    \end{itemize}

 \section{Network-wide Configuration Synthesis}
\label{sec:overview}

We now illustrate our configuration synthesis approach on a simple example.
We highlight how, given a network and a set of requirements, we can pose the synthesis problem as an instance of input synthesis for stratified Datalog.

\subsection{Motivating Example}
We consider the simple network topology, depicted in Figure~\ref{fig:overview}(b), composed of $4$ routers denoted $A$, $B$, $C$ and $D$. Routers $B$ and $C$ can reach the external network \pred{Ext},
and router $D$ is directly connected to two internal networks \pred{N1} and \pred{N2}.
In the following, we use the term traffic class to refer to a set of packets (e.g. packets destined to \pred{N1}) that are handled analogously according to the requirements. In practice, each traffic class may contain thousands of IP prefixes~\cite{miningpolicies}.

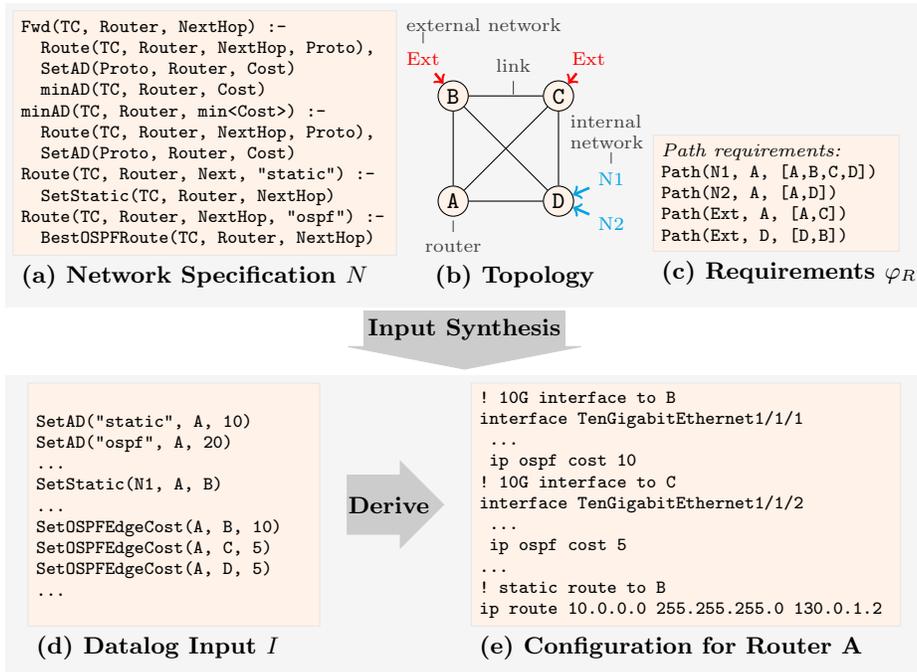
\begin{figure}[t!]
	\begin{tikzpicture}
	
	\node[anchor=north west, fill=gray!08, draw=none, minimum width=1\textwidth, minimum height=115pt] (input) at (0,0) {};

	\node[anchor=north west, fill=Peach!10, draw=gray!20, minimum height=92pt] (spec) at (0.1,-0.1) {
\begin{minipage}{0.4\textwidth}
\scriptsize
\begin{verbatim}
Fwd(TC, Router, NextHop) :-
  Route(TC, Router, NextHop, Proto),
  SetAD(Proto, Router, Cost)
  minAD(TC, Router, Cost)
minAD(TC, Router, min<Cost>) :-
  Route(TC, Router, NextHop, Proto),
  SetAD(Proto, Router, Cost)
Route(TC, Router, Next, "static") :-
  SetStatic(TC, Router, NextHop)
Route(TC, Router, NextHop, "ospf") :-
  BestOSPFRoute(TC, Router, NextHop)
\end{verbatim}
\end{minipage}};
\node[anchor=north west] (spec-label) at (spec.south west) {\bf \small (a) Network Specification $N$};

		\small
	\node[anchor=south west] (net) at ($(spec-label.south west) + (5.5, -0)$) {\bf \small (b) Topology};
	\node[anchor=west, circle, fill=Peach!10, draw=black, inner sep=1.5pt] (B) at ($(net.west) - (-0.15,-2.4)$) {\verb|B|};
	\node[circle, fill=Peach!10, draw=black, inner sep=1.5pt] (A) at ($(B) - (0, 1.4)$) {\verb|A|};
	\node[circle, fill=Peach!10, draw=black, inner sep=1.5pt] (C) at ($(B) + (1.4, 0)$) {\verb|C|};
	\node[circle, fill=Peach!10, draw=black, inner sep=1.5pt] (D) at ($(C) - (0, 1.4)$) {\verb|D|};
	
	\path [-, draw=black]
	(B) edge (C)
	(A) edge (B)
	(A) edge (D)
	(C) edge (D)
	(A) edge (C)
	(B) edge (D);

	\node[color=red] (E1) at ($(B) + (-0.4, 0.5)$) {\scriptsize Ext};	\node[color=red] (E2) at ($(C) + (0.4, 0.5)$) {\scriptsize Ext}; 	\path[->, draw=red, line width=1pt] (E1) edge (B);
	\path[->, draw=red, line width=1pt] (E2) edge (C);
	
	\node[color=Cerulean] (I1) at ($(D) + (0.7, 0.3)$) {\scriptsize N1}; 	\node[color=Cerulean] (I2) at ($(D) + (0.7, -0.3)$) {\scriptsize N2}; 	\path[->, draw=Cerulean, line width=1pt] (I1) edge (D);
	\path[->, draw=Cerulean, line width=1pt] (I2) edge (D);
	
	\node[color=darkgray] (router) at ($(A) + (0,-0.6)$) {\scriptsize router};
	\path[-, draw=darkgray] (router.north) edge ($(A.south) - (0,0.05)$);
	
	\node[color=darkgray] (link) at ($(B) + (0.8, 0.4)$) {\scriptsize link};
	\path[-, draw=darkgray] (link.south) edge ($(link.south) + (0, -0.18)$);
	
	\node[anchor=west, color=darkgray] (ext) at ($(E1.west) + (0, 0.45)$) {\scriptsize external network};
	\path[-, draw=darkgray] ($(E1.north) + (0, 0.15)$) edge (E1.north);
	
	\node[anchor=west, color=darkgray] (int) at ($(D.west) + (0.2, 0.9)$) {\scriptsize \begin{tabular}{l}internal\\ network\end{tabular}};
	\path[-, draw=darkgray] ($(I1.north) + (-0, 0.15)$) edge ($(I1.north) - (0,0)$);

	\node[fill=Peach!10, draw=gray!20, anchor=north west, minimum height=44pt] (reqs) at ($(spec.north east) + (3.4,-1.655)$) {
	\begin{minipage}{0.25\textwidth}
	\scriptsize
	\raggedright {\em Path requirements:}\\
	\texttt{Path(N1, A, [A,B,C,D])\\
	Path(N2, A, [A,D])\\
	Path(Ext, A, [A,C])\\
 	Path(Ext, D, [D,B])}\\
	\end{minipage}	};
	\node[anchor=north west] at (reqs.south west) {\bf \small (c) Requirements $\varphi_R$};

	\node[anchor=north west, fill=gray!08, draw=none, minimum width=1\textwidth, minimum height=114pt] (output) at ($(input.south west) + (0,-0.9)$) { };
	
	\node[draw=gray!20, fill=Peach!10, anchor=north west, minimum height=92pt] (tuples) at ($(output.north west) + (0.3,-0.1)$) {
		\begin{minipage}{0.3\textwidth}
		\scriptsize
		\begin{verbatim}
		SetAD("static", A, 10)
		SetAD("ospf", A, 20)	
		...
		SetStatic(N1, A, B)
		...
		SetOSPFEdgeCost(A, B, 10)
		SetOSPFEdgeCost(A, C, 5)
		SetOSPFEdgeCost(A, D, 5)
		...
		\end{verbatim}
		\end{minipage}
	};
	\node[anchor=north west] at (tuples.south west) {\bf \small (d) Datalog Input $I$};

	\node[draw=gray!20, anchor=north west, fill=Peach!10, minimum height=92pt] (config) at ($(tuples.north east) + (2,0)$) {		
\begin{minipage}{0.45\textwidth}
\scriptsize
\begin{verbatim}
! 10G interface to B
interface TenGigabitEthernet1/1/1
 ...
 ip ospf cost 10
! 10G interface to C
interface TenGigabitEthernet1/1/2
 ...
 ip ospf cost 5
...
! static route to B
ip route 10.0.0.0 255.255.255.0 130.0.1.2
\end{verbatim}
\end{minipage}
	};	
	\node[anchor=north west] at (config.south west) {\bf \small (e) Configuration for Router A};

	\node[fill=gray!40, inner sep=2pt, minimum height=0pt] (synth) at ($(input.south) - (0, 0.3)$) {\bf \!\begin{tabular}{l}Input Synthesis\end{tabular}\!};
	
	\draw[fill=gray!40, draw=gray!40] ($(synth.south west) - (0.2, -0.05)$) node[anchor=north]{ }
	-- ($(synth.south east) + (0.2, 0.05)$) node[anchor=north] {}
	-- ($(synth.south) + (0, -0.25)$) node[anchor=south]{ }
	-- cycle;
	
	\draw[-, color=gray!40, line width=1.8mm, -triangle 90,postaction={draw, line width=8mm, shorten >=4mm, -}] ($(tuples.east) + (0.35,0)$) -- ($(config.west) - (0.35,0)$);
	\node at ($(tuples.east) + (0.85,0)$) {\bf \ \ Derive\ \ };
	
	\end{tikzpicture}
	\caption{Network-wide Configuration Synthesis. The input is: {\bf (a)} declarative network specification~$N$ in stratified Datalog
				{\bf (b)} network topology, 		and {\bf (c)} routing requirements~$\varphi_R$.
				The output is: {\bf (d)} a Datalog input~$I$ that results in a forwarding state satisfying the requirements. Configurations {\bf (e)} are derived from~$I$.}
	\label{fig:overview}
\end{figure}
 
\para{Computation of Forwarding State}
We first informally describe how each router's forwarding entries are computed, assuming the configuration is provided.

Each router runs both, OSPF and BGP protocols, and in addition can also be configured with static routes. The computation of OSPF is based on finding least-cost paths to the internal destinations as well as to all routers in the network, where cost is the sum of the link weights defined in router configurations. The least-cost paths are then used to generate forwarding entries at each router to all internal destinations. In our example, these internal destinations are \pred{N1} and \pred{N2}. In contrast, BGP computes forwarding entries to reach external destinations, \pred{Ext} in our example. The computed forwarding entries define the next hop router for each destination. For example, BGP computes an entry at router $A$ for \pred{Ext} which forwards packets to a border router (i.e., either $B$ or $C$).
To decide which router the entry should forward to, each BGP router selects the egress point (i.e., border router) to reach an externally-learned prefix based on a preference value. This preference is (typically) defined in the configuration of each border router and propagated network-wide. If multiple routers announce the same preference for a prefix, internal BGP routers directs traffic to the closest egress point, according to the OSPF costs.

Once BGP and OSPF have finished computing their forwarding entries, each router takes these entries (along with those defined via static routes) and selects the OSPF-, BGP-, static route- produced forwarding entry with the highest preference (in networking terms, higher preference means lower administrative cost) defined in its local configuration. The union of all forwarding entries obtained at the routers is referred to as the forwarding state of the network.

\para{Configuration Synthesis}
Next, we illustrate the opposite direction (and one this work focuses on): given requirements $\varphi_R$, find a configuration which the protocols use to compute a forwarding state (as described above) that satisfies~$\varphi_R$.

Let us consider the four path requirements given in Figure~\ref{fig:overview}(c). The first two state that $A$ must forward packets for the traffic classes \pred{N1} and \pred{N2} along the paths $A \DashedArrow[solid] B \DashedArrow[solid] C \DashedArrow[solid] D$ and $A\DashedArrow[solid]D$, respectively.
Note that these two requirements might reflect a security policy in the network or generated by a traffic engineering optimization tool~\cite{rfc3272, fortz2002traffic}.
These two requirements cannot be enforced using OSPF alone. The reason is that, as discussed, OSPF works by selecting the least-cost path (by summing the weights on the links) and there is no assignment of weights to links which would lead to least-cost paths that exactly match the two path requirements.

Yet, the two requirements can be enforced by: {\em (i)} generating a static route- based forwarding entry at $A$ to forward packets for \pred{N1} to $B$; {\em (ii)} configuring link weights so paths $A \DashedArrow[solid] D$ and $B \DashedArrow[solid] C \DashedArrow[solid] D$ have the lowest OSPF costs from $A$ to $D$ and, respectively, from $B$ to $D$; and {\em (iii)} on router $A$, configure a higher preference for forwarding entries based on static routes than OSPF forwarding entries. Because a static route forwarding entry is only generated for destination \path{N1} (from {\em(i)}) and not \path{N2}, this means the entry for \pred{N1} will forward the traffic to router $B$ while the entry for \pred{N2} will be the OSPF generated one (from {\em(ii)}).

The last two path requirements state that $A$ and $D$ must forward packets destined to the traffic class \pred{Ext} to $C$ and $B$, respectively. The two path requirements can by satisfied by: {\em (i)} setting identical BGP router preferences at the local configurations of $B$ and $C$; and {\em (ii)} configuring link weights so that paths $A \DashedArrow[solid] C$ and $D \DashedArrow[solid] B$ have the lowest costs from $A$ to $C$ and from $D$ to $B$, respectively. In this way, BGP will use the results from the OSPF least-cost paths to compute its forwarding entries to \pred{Ext}. This is an example where BGP interacts with OSPF and uses information from its computation.

The following is the final configuration produced by our synthesizer (the synthesizer is discussed in later sections): 
\begin{itemize} 
\item weight $10$ is assigned to link $A \DashedArrow[solid] B$,
\item weight $5$  is assigned to links $B \DashedArrow[solid]C$, $C \DashedArrow[solid ]D$, and $A \DashedArrow[solid] C$,
\item weight $4$  is assigned to link $D \DashedArrow[solid] B$,
\item weight $100$ is assigned to the remaining links,
\item a static route- based forwarding entry is defined at router $A$ to forward traffic for $N1$ to $B$, and
\item the router preference for all routers is set to $100$.
\end{itemize}
In Figure~\ref{fig:overview}(e), we illustrate an excerpt of router $A$'s local configuration.

\para{Phrasing the Problem as Inputs Synthesis for Stratified Datalog}
A key insight of our work is to pose the question of finding a network configuration as an instance of input synthesis for stratified Datalog.

First, we declaratively specify the behavior of the network, i.e. the distributed protocols that the routers run, the protocol interactions, and the network topology, as a stratified Datalog program~$N$.
As requirements usually pertain to the stable forwarding state, the stratified Datalog encoding captures the stable state of these routing protocols as opposed to intermediate computation steps.
Few relevant Datalog rules are given in Figure~\ref{fig:overview}(a); we detail this specification step in Section~\ref{sec:specification}.
The resulting Datalog program derives a predicate \pred{Fwd} that defines the forwarding entries computed by all routers, where \pred{Fwd(TC, Router, NextHop)} is derived if \pred{Router} forwards packets for traffic class \pred{TC} to router \pred{NextHop}.

Second, we can directly express routing requirements as constraints over the predicate \pred{Fwd}. We denote these constraints with $\reqs$ in Figure~\ref{fig:overview}.

Finally, an input~$I$ to the Datalog program~$N$ identifies a network-wide configuration.
We formalize the network-wide configuration synthesis problem as:

\begin{definition}
	The network-wide configuration synthesis problem is:\\
	\begin{tabular}{lp{303pt}}
		{\bf Input} & A declarative network specification~$N$ and routing requirements $\varphi_R$.
						\\
		{\bf Output} & A Datalog input~$I$ such that
				$\sem{N}_I \models \reqs$, if such an input exists; otherwise, return \unsat.
	\end{tabular}
\end{definition}

In our definition, $\sem{N}_I$ denotes the fixed point of the Datalog program~$N$ for the input~$I$, and $\sem{N}_I \models \reqs$ holds if this fixed point satisfies the constraints $\reqs$.

Synthesizing inputs for stratified Datalog is, however, a difficult (and, in general, undecidable)  problem~\cite{Halevy:2001:SAD:502102.502104}.
The problem is, however, decidable if we fix a finite set of values to bound the set of inputs.
This is reasonable in the context of networks, where values represent finitely many routers, interfaces, and configuration parameters.

To address the problem, we introduce a new iterative synthesis algorithm that partitions the Datalog program~$P$ into strata~$P_1, \ldots, P_n$,
finds an input~$I_i$ for each stratum~$P_i$ and then construct an input~$I$ for the Datalog program~$P$.
Each stratum~$P_i$ is a semi-positive Datalog program that enjoys the property that if a predicate is derived by the rules after some number of steps, then it must be contained in the fixed point.
We describe this algorithm in Section~\ref{sec:datalog-synthesis}.

%Finally, to make our input synthesis algorithm work on practical network examples, we present network-specific optimizations, including network-specific constraints for reducing the space of possible configurations and program simplifications. We describe these in Section~\ref{sec:network-synthesis}. 
 \section{Background: Stratified Datalog}

We briefly overview the syntax and semantics of stratified Datalog.

\para{Syntax}
Datalog's syntax is given in Figure~\ref{fig:datalog}. We use $\overline{r}$, $\overline{l}$, and $\overline{t}$ to denote zero or more rules, literals, and terms separated by commas, respectively.
A Datalog program is {\em well-formed} if for any rule $a\leftarrow \overline{l}$, we have $\getvars{a}\subseteq \getvars{\overline{l}}$, where $\getvars{\overline{l}}$ returns the set of variables in $\overline{l}$.

A predicate is called {\em extensional} if it appears only in the bodies of rules (right side of the rule), otherwise (if it appears at least once in a rule head) it is called {\em intensional}. We denote the sets of extensional and intensional predicates of a program $P$ by $\edb{P}$ and $\idb{P}$, respectively.

A Datalog program~$P$ is {\em stratified} if its rules can be partitioned into strata $P_1, \ldots, P_n$ such that if a predicate $p$ occurs in a positive (negative) literal in the body of a rule in $P_i$, then all rules with $p$ in their heads are in a stratum $P_j$ with $j \leq i$ ($j< i$). Stratification ensures that predicates that appear in negative literals are fully defined in lower strata.

We syntactically extend stratified Datalog with aggregate functions such as \pred{min} and \pred{max}. This extension is possible as stratified Datalog is equally expressive to Datalog with stratified aggregate functions; for details see~\cite{Mumick1995}

\begin{figure}[t]
	\centering
	\setlength{\tabcolsep}{2pt}
	\renewcommand{\arraystretch}{1}
	\begin{tabular}{rrclrrclrrcl}
		{\em (Program)} & $P$ & $::=$ & $\overline{r}$ \hspace{20pt} 
		& {\em (Literal)} & $l$ & $::=$ & $a\mid \neg a$ 
		& {\em (Variables)} & $X, Y$ & $\in $ & \vars\\
				{(\em Rule)} & $r$ & $::=$ & $a \leftarrow \overline{l}$
		& {\em (Predicates)} &$p, q$ & $\in$ & \preds
		&  {\em (Values)} & $v$ & $\in$ & \vals\\
				{(\em Atom)} & $a$ & $::=$ & $p(\overline{t})$ 
		& {\em (Term)} & $t$ & $::=$ & $X\mid v$ \hspace{14pt} \\		
	\end{tabular}\\
	\caption{Syntax of stratified Datalog}
	\label{fig:datalog}
\end{figure} 

\para{Semantics}
Let ${\cal A} = \{ p(\overline{t})\mid \overline{t}\subseteq \vals \}$ denote the set of all ground (i.e. variable-free) atoms.
The complete lattice $({\cal P}({\cal A}), \subseteq, \cap, \cup, \emptyset, {\cal A})$ partially orders the set of interpretations ${\cal P}(\cal A)$.

Given a substitution $\sigma\in \textit{Vars}\to {\it Vals}$ mapping variables to values. Given an atom~$a$, we will write $\sigma(a)$ for the ground atom obtained by replacing the variables in $a$ according to $\sigma$; e.g., $\sigma(p(X))$ returns the ground atom $p(\sigma(X))$.
The consequence operator $T_P\in {\cal P}({\cal A})\to {\cal P}({\cal A})$ for a program~$P$ is defined as
\[
T_P(A) = A\cup \{\sigma(a)\mid a\leftarrow l_1\ldots l_n \in P, \forall l_i\in \overline{l}.\ A\vdash \sigma(l_i)\}
\]
where $A \vdash \sigma(a)$ if $\sigma(a)\in A$ and $A \vdash \sigma(\neg a)$ if $\sigma(a)\not\in A$.

An input for $P$ is a set of ground atoms constructed using $P$'s extensional predicates.
Let $P$ be a program with strata $P_1, \ldots, P_n$ and $I$ be an input for $P$. The model of $P$ for $I$, denoted by $\sem{P}_I$, is $M_n$, where $M_0 = I$ and $M_i=\bigcap \{A\in {\sf fp}\ T_{P_i}\mid A\subseteq M_{i-1}\}$ is the smallest fixed point of $T_{P_i}$ that is greater than the lower stratum's model $M_{i-1}$.

\section{Declarative Network Specification}\label{sec:specification}

In this section, we first describe how we declarative specify the behavior of the network as a Datalog program. Afterwards, we discuss how routing requirements are specified as constraints over the Datalog program's fixed point.

\subsection{Specifying Networks}\label{sec:spec-networks}

To faithfully capture a network's behavior, we model
{\em (i)}~the behavior of routing protocols and their interactions
and 
{\em (ii)}~the topology of the network.

\para{Expressing Protocols in Stratified Datalog}
We formalize individual routing protocols and how routers combine the forwarding entries computed by these protocols as a stratified Datalog program~$N$.
The Datalog program~$N$ derives the predicate \pred{Fwd(TC, Router, NextHop)}, which represents the network's global forwarding state. 
In Figure~\ref{fig:overview}(a), for example, we show the relevant rules that define how the forwarding entries computed by OSPF are combined with those defined via static routes.
The predicate \pred{Route(TC, Router, NextHop, Proto)} captures the forwarding entries of OSPF and static routes.
The top Datalog rule states that routers select, for each traffic class~\pred{TC}, the forwarding entry with the minimal administrative cost (\pred{minAD}) calculated over all protocols via the second Datalog rule in Figure~\ref{fig:overview}(a). 
The bottom two rules define the predicate \pred{Route}, which collects the forwarding entries defined via static routes and computed by OSPF. 
We remark that OSPF routes (represented by the predicate \pred{BestOSPFRoute}) are defined through additional Datalog rules that capture the behavior of the OSPF protocol\footnote{\ifarxiv
	 A detailed OSPF model can be found in  Appendix~\ref{sec:ospf-specification}.
\else
	 A detailed OSPF model can be found in  the technical report~\cite{synet-arxiv}.
\fi}.

\para{Network Topology}
The network topology is also captured via Datalog rules in the program~$N$.
We model each router as a constant and use predicates to represent the topology. For example, the predicate \pred{SetLink(R1, R2)}
represents that two routers $R1$ and $R2$ are connected via a link, and we add the Datalog rule $\pred{SetLink(R1, R2)}\leftarrow \pred{true}$
to define such a link.

\subsection{Specifying Requirements}\label{sec:requirements}

We specify the requirements as function-free first-order constraints over the predicate \pred{Fwd(TC, Router, NextHop)}, which defines the network's forwarding state.
We write $A\models \varphi$ to denote that a Datalog interpretation~$A$ satisfies~$\varphi$.
For illustration, we describe how common routing requirements can be specified:
\begin{description}
	\item[\pred{Path(TC, R1, [R1, R2, .., Rn])}] (Path requirement): 
	packets for traffic class \pred{TC} must follow the path $\pred{R1, .., Rn}$. These requirements are specified as a conjunction over the predicate \pred{Fwd}.
	\item[$\forall \pred{R1, R2}.\ \pred{Fwd(TC1, R1, R2)} \Rightarrow \neg \pred{Fwd(TC2, R1, R2)}$] (Traffic isolation): the paths for two distinct traffic classes~$\pred{TC1}$ and $\pred{TC2}$ do not share links in the same direction.
	\item[$\pred{Reach(TC, R1, R2)}$] (Reachability): packets for traffic class~\pred{TC} can reach router \pred{R2} from router \pred{R1}. The predicate \pred{Reach} is the transitive closure over the predicate \pred{Fwd} (defined via Datalog rules).
	\item[$\forall \pred{TC}, \pred{R}.\ (\neg \pred{Reach(TC, R, R)})$] (Loop-freeness): generic requirement stipulating that the forwarding plane has no loops.
\end{description}

More complex requirements, such as way pointing, can be specified based on the core function-free first-order constraints provided by \tool{}. Further, \tool{} can be used as a backend for a high-level requirements language that is easier to use by a network operator.

\subsection{Network-wide Configurations}
\label{sec:protocol-configs}
The input protocol configurations deployed at the network's routers are represented as input \textit{edb} predicates to the Datalog programs that formalize the protocols. For example, the local OSPF configuration for a router specifies the weights associated with the links connected to that router; this is represented by the \textit{edb} predicate \pred{SetOSPFEdgeCost(Router, NextHop, Weight)}.

A subset of the synthesized Datalog input for our motivating example is given in Figure~\ref{fig:overview}(d).
Here, \pred{SetAD} defines the administrative cost of static routes to be lower than that of OSPF (so static routes are prefered over forwarding entries computed by OSPF).
The predicate \pred{SetStatic(N1, A, B)}, which represents static routes, defines a static route for~\pred{N1} from $A$ to $B$.
The predicate \pred{SetOSPFEdgeCost} defines the links' weights.

  \section{Input Synthesis for Stratified Datalog}\label{sec:datalog-synthesis}

We now present a new iterative algorithm for synthesizing inputs for stratified Datalog. We first describe the high-level flow of the algorithm before presenting the details.

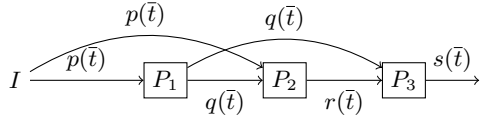
\begin{figure}[t!]
	\centering
	\begin{tikzpicture}
	
	\node[draw=black] (P1) at (0,0) {$P_1$};
	\node (I) at ($(P1) - (2,0)$) {$I$};
	\node[draw=black, anchor=west] (P2) at ($(P1.east) + (1,0)$) {$P_2$};		
	\node[draw=black, anchor=west] (P3) at ($(P2.east) + (1,0)$) {$P_3$};		
	
	\draw[->] (I) edge node[above] {$p(\overline{t})$} (P1);	
	\draw[->] (P1) edge node[below] {$q(\overline{t})$} (P2);
	\draw[->, bend left] (P1) edge node[above] {$q(\overline{t})$} (P3);
	\draw[->, bend left] (I) edge node[above] {$p(\overline{t})$} (P2);
	\draw[->] (P2) edge node[below] {$r(\overline{t})$} (P3.west);	
	\draw[->] (P3) edge node[above] {$s(\overline{t})$} ( $(P3) + (1,0)$);	
	\end{tikzpicture}
	\caption{A Datalog program with strata $P_1$, $P_2$, and $P_3$, and flow of predicates between the strata.}
	\vspace{14pt}
	\label{fig:eager-example}
\end{figure}

\paragraph{High-Level Flow}

Consider the stratified Datalog program with strata $P_1$, $P_2$, and $P_3$, depicted in Figure~\ref{fig:eager-example}. 
Incoming and outgoing edges of a stratum~$P_i$ indicate the \textit{edb} predicates and, respectively, the \textit{idb} predicates of that stratum. For example, the stratum~$P_3$ takes as input predicates $q(\overline{t})$ and $r(\overline{t})$ and derives the predicate $s(\overline{t})$. 
Our iterative algorithm first synthesizes an input~$I_3$ for $P_3$ which determines the predicates $q(\overline{t})$ and $r(\overline{t})$ that $P_1\cup P_2$ must output. To synthesize such an input for a single stratum, we present an algorithm, called \synthbox, that addresses the input synthesis problem for semi-positive Datalog programs~\cite[Chapter~15.2]{Abiteboul:1995:FDL:551350}, i.e. Datalog programs where negation is restricted to \textit{edb} predicates. 
After synthesizing an input~$I_3$ for $P_3$, our iterative algorithm synthesizes an input~$I_2$ for $P_2$ such that the fixed-point $\sem{P_2}_{I_2}$ produces the predicates $r(\overline{t})$ that are contained in the already synthesized input~$I_3$ for $P_3$.
We note that this iterative process may require backtracking, in case no input for $P_2$ can produce the desired predicates $r(\overline{t})$ contained in $I_3$.
The algorithm terminates when it synthesizes inputs for all three strata.

In the following, we first present the algorithm \synthbox that is used to synthesize an input for a single stratum (which is a semi-positive program). 
Then, we present the general algorithm, called \synth, that iteratively applies \synthbox for each stratum to synthesize inputs for stratified Datalog programs.

\subsection{Input Synthesis for Semi-positive Datalog with SMT}

\label{sec:semipositive-datalog-synthesis}

The key idea is to reduce the input synthesis problem to satisfiability of SMT constraints: Given a semi-positive Datalog program~$P$ and a constraint~$\varphi$, we encode the question $\exists I.\ \sem{P}_I\models \varphi$ into an SMT constraint $\psi$. If $\psi$ is satisfiable, then from a model of $\psi$ we can derive an input~$I$ such that $\sem{P}_I\models \varphi$.

\para{SMT Encoding Challenges}
Given a Datalog program~$P$ and a constraint $\varphi$,
encoding the question $\exists I.\ \sem{P}_I\models \varphi$ with SMT constraints is non-trivial due to the mismatch between Datalog's program fixed point semantics and the classical semantics of first-order logic.
This means that simply taking the conjunction of all Datalog rules into an SMT solver does not solve our problem.
For example, consider the following Datalog program~$\tcprog$:
\[
\begin{array}{rcl}
	tc(X, Y) & \leftarrow & e(X, Y)\\
	tc(X, Y) & \leftarrow & tc(X, Z), tc(Z,Y)\\
\end{array}
\]
which computes the transitive closure of the predicate $e(X, Y)$. A naive way of encoding these Datalog rules with SMT constraints:
\[
\begin{array}{rcl}
\forall X, Y.\ (e(X, Y) & \Rightarrow & tc(X, Y))\\
\forall X, Y.\ (( \exists Z.\ tc(X, Z) \wedge tc(Z,Y)) & \Rightarrow & tc(X, Y)) \\
\end{array}
\]
and we denote the conjunction of these two SMT constraints as~$[\tcprog]$.
Now, suppose we 
have the fixed point constraint 
$\tcprop = (\neg e(v_0, v_2))\wedge tc(v_0, v_2)$ and we
want to generate an input $I$ so that $\sem{\tcprog}_I\models \tcprop$.
A model that satisfies $[\tcprog]\wedge \tcprop$ is
\[
\model = \{e(v_0, v_1), tc(v_0, v_1), tc(v_0, v_2)\}
\]
The input derived from this model, obtained by projecting $\model$ over the \textit{edb} predicate~$e$, is $I_\model = \{e(v_0, v_1)\}$. We get
\[
\sem{\tcprog}_{I_\model} = \{ e(v_0, v_1), tc(v_0, v_1) \}
\]
and so $\sem{\tcprog}_{I_\model}\not\models \tcprop$, which is clearly not what is intended.

\para{SMT Encoding}
Our key insight is to split the constraint~$\varphi$ into a conjunction of positive and negative clauses, where
a clause $\varphi$ is positive (resp., negative) if $A\models \varphi$ implies that $A'\models \varphi$ for any interpretation $A'\supseteq A$ (resp., $A'\subseteq A$).
We can then unroll recursive predicates to obtain a sound encoding for positive constraints, and we do not unroll them to get a sound encoding for negative constraints.

\begin{figure}[t!]
	\small
	\centering
	\renewcommand*{\arraystretch}{1}		
	\[
	\begin{array}{lll}
	    \translate{P}{k} & = & \bigwedge\limits_{p\in \idb{P}} \negencode{P}{p} \wedge \posencode{P}{p}{k}\\
	    		    \negencode{P}{p} & = & \bigwedge\limits_{p(\overline{X})\leftarrow \overline{l}\in P}\   
		    \forall \overline{X}.\ \big( (\exists \overline{Y}.\ \bigwedge \overline{l}) \Rightarrow p(\overline{X}) \big), \text{where}\ \overline{Y} = \getvars{\overline{l}}\setminus \overline{X}\\
    	\posencode{P}{p}{k} & = & \bigwedge\limits_{0 < i\leq k} \step{P}{p}{i}\\
    	\step{P}{p}{i} & = & \forall \overline{X}.\ \big(p_i(\overline{X}) \Leftrightarrow ( \bigvee\limits_{p(\overline{X})\leftarrow \overline{l}\in P} \exists \overline{Y}.\ \tau(\overline{l}, i-1))\big), \text{where}\ \overline{Y} = \getvars{\overline{l}}\setminus \overline{X}\\
   		\tau(\overline{l}, k) &  = & 
   		\left\{
   		\begin{array}{ll}
   		\tau(l_1, k) \wedge \cdots \tau(l_n, k) &\ \text{if}\ \overline{l} = l_1\wedge \cdots \wedge l_n\\
   		\neg \tau(p(\overline{t}), k) &\ \text{if}\ \overline{l} = \neg p(\overline{t})\\
   		{\sf false}  &\ \text{if}\ \overline{l} =p(\overline{t}), p\in \idb{P}, k = 0\\
   		p_k(\overline{t}) &\ \text{if}\ \overline{l} =p(\overline{t}), p\in \idb{P}, k > 0\\
   		l &\ \text{otherwise}\\
   		\end{array}
   		\right.
	\end{array}
	\]
	\vspace{-14pt}
	\caption{Encoding a Datalog program~$P$ with constraints~$[P]_k$}
	\label{fig:encoding}
\end{figure} 
The encoding of a Datalog program~$P$ into an SMT constraint is defined in Figure~\ref{fig:encoding}.
The resulting SMT constraint is denoted by $[P]_k$, where the parameter~$k$ defines the number of unroll steps.
In the encoding we assume that {\em (i)} all terms in rules' heads are variables and {\em (ii)} rules' heads with the same predicate have identical variable names. Note that any Datalog program can be converted into this form using rectification~\cite{Ullman89} and variable renaming.

\para{Function \textsc{Encode}}
The constraint returned by $\textsc{Encode}(p, P)$ states that an atom $p(X)$ is derived if $P$ has a rule that derives $p(\overline{X})$ and whose body evaluates to true.
To capture Datalog's semantics,
the variables in $p(\overline{X})$ are universally quantified, while those in the rules' bodies are existentially quantified.
This constraint $\textsc{Encode}(p, P)$ is sound for negative requirements, but not for positive ones as it does not state that  $p(\overline{X})$ is derived {\em only if} a rule body with $p(\overline{X})$ in the head evaluates to true.

\para{Functions \textsc{Unroll} and \textsc{Step}}
The constraint returned by $\textsc{Step}(P, p, i)$ encodes whether an atom $p(X)$ is derived after $i$ applications of $P$'s rules;
e.g., $p(X)$'s truth value  after $3$ steps is represented with the atom $p_3(X)$.
Intuitively, $p(X)$ is true iff 
there is a rule that derives $p(X)$ and whose body evaluates to true using the atoms derived in previous iterations.
Which atoms are derived in previous iterations is captured by the literal renaming function~$\tau$.
Note that $\tau(l, 0)$ returns $\sf false$ for any {\em idb} literal~$l$ since all intensional predicates are initially $\sf false$. Further, $\tau(l, k)$ returns $l$ for any extensional literal $l$ (the case ``otherwise'' in Figure~\ref{fig:encoding}) since their truth value does not change. Finally, the constraint returned by $\textsc{Unroll}(P, p, k)$ conjoins $\textsc{Step}(P, p, 0)$, \ldots,  $\textsc{Step}(P, p, k)$ to capture the derivation of $p(X)$ after $k$ steps. This is sound for positive requirements, but not for negative ones since more $p(X)$ atoms may be derived after $k$ steps.

\para{Example}
To illustrate the encoding, we translate the Datalog program:
\[
\begin{array}{rcl}
tc(X, Y) & \leftarrow & e(X, Y)\\
tc(X, Y) & \leftarrow & tc(X, Z), tc(Z,Y)\\
\end{array}
\]
which computes the transitive closure of the predicate $e(X, Y)$.
This program has one {\em idb} predicate, $tc$. The function $\textsc{Encode}(P, tc)$ returns
\[
\begin{array}{rcl}
(\forall X, Y.\ e(X, Y) & \Rightarrow & tc(X, Y))\\
\wedge
(\forall X, Y.\ (\exists Z.\ tc(X, Z) \wedge tc(Z,Y)) & \Rightarrow & tc(X,Y))\\
\end{array}
\]
We apply function $\textsc{Unroll}(a, P, 2)$ for $k=2$, which after simplifications returns
\[
\begin{array}{rcl}
\forall X, Y.\ (tc_1(X, Y) & \Leftrightarrow & e(X, Y)) \\
\forall X, Y.\ (tc_2(X, Y) & \Leftrightarrow & e(X, Y) \vee (\exists Z.\ tc_1(X, Z) \wedge tc_1(Z, Y) ) \\
\end{array}
\]
In the constraints, the predicates $tc_1$  and $tc_2$ encode the derived predicates $tc$ after $1$ and, respectively, $2$, derivation steps.

\newcommand{\rewrite}[2]{\textsc{Rewrite}(#1, #2)}
\newcommand{\simplify}[1]{\textsc{Simplify}(#1)}

\begin{algorithm}[t!]
	\small
	\DontPrintSemicolon
	\KwIn{Semi-positive Datalog program~$P$ and a constraint~$\varphi$}
	\KwOut{An input $I$ such that $\sem{P}_I \models \varphi$ or $\bot$}
	\Begin{
			$\varphi'\gets \simplify{\varphi}$\;
			\For{$k \in [1.. \textit{bound}_k]$}{
				$\varphi_k \gets \rewrite{\varphi'}{k}$\;
				$\psi \gets \translate{P}{k} \wedge \varphi_k$\;
				\If{$\exists J.\ J\models \psi$}{
					$I \gets \{p(\overline{t})\in J\mid p\in \edb{P}\}$, where $J\models \psi$\;
					\Return{$I$}\;
				}
			}
			\Return{$\bot$}	
	}
\caption{Algorithm~$\synthbox$ for semi-positive Datalog}
\label{alg:semipositive-input-synthesis}
\end{algorithm}

\para{Algorithm}
Algorithm $\synthbox(P, \varphi)$, given in Algorithm~\ref{alg:semipositive-input-synthesis}, first calls function $\simplify{\varphi}$ that {\em (i)} instantiates any quantifiers in~$\varphi$ and \emph{(ii)} transforms the result into a conjunction of clauses, where each clause is a disjunction of literals. 

Then, the algorithm iteratively unrolls the Datalog rules, up to a pre-defined bound, called $\textit{bound}_k$.
In each step of the for-loop, the algorithm generates an SMT constraint that captures \emph{(i)} which atoms are derived after $k$ applications of $P$'s rules and \emph{(ii)} which atoms are never derived by $P$. The resulting SMT constraint is denoted by $\translate{P}{k}$.
The algorithm also rewrites the simplified constraint $\varphi'$ using the function $\rewrite{\varphi'}{k}$ which recursively traverses conjunctions and disjunctions in the simplified constraint $\varphi'$ and maps positive literals to the $k$-unrolled predicate  $p_k(\overline{t})$ and negative literals to $\neg p(\overline{t})$:
\[
\rewrite{\varphi}{k}\! =\!
\left\{
\begin{array}{ll}
p_k(\overline{t}) & \text{if}\ \varphi\! =\! p(\overline{t})\\
\neg p(\overline{t}) & \text{if}\ \varphi\! =\! \neg p(\overline{t})\\
\rewrite{\varphi_1}{k} \vee \cdots \vee \rewrite{\varphi_n}{k} & \text{if}\ \varphi\! =\! \varphi_1\vee .. \vee \varphi_n\\
\rewrite{\varphi_1}{k} \wedge \cdots \wedge \rewrite{\varphi_n}{k} & \text{if}\ \varphi\! =\! \varphi_1\wedge .. \wedge \varphi_n\\
\end{array}\right.
\]
Note that since $\vee$ and $\wedge$ are monotone, negative literals constitute negative constraints and positive literals constitute positive constraints.
 
If the resulting constraint~$\translate{P}{k}\wedge \psi_k$ is satisfiable, then an input is derived by projecting the interpretation $I$ that satisfies the constraint over all \textit{edb} predicates.
Note that if there is an input~$I$ such that $\sem{P}_I\models \varphi$ and for which the fixed point~$\sem{P}_I$ is reached in less than $\textit{bound}_k$ steps, then $\synthbox(P, \varphi)$ is guaranteed to return an input. 

\begin{theorem}
	Let $P$ be a semi-positive Datalog program, $\varphi$ a constraint.\\If
	$\synthbox(P, \varphi) = I$ then $\sem{P}_I\models \varphi$. \footnote{\ifarxiv
	The theorem's proof can be found in Appendix~\ref{sec:proofs}.
\else
	The theorem's proof can be found in the technical report~\cite{synet-arxiv}.
\fi}
\end{theorem}

\newcommand{\backtrack}{\textit{backtrack}}

\subsection{Iterative Input Synthesis for Stratified Datalog}

\label{sec:stratified-datalog-synthesis}

\begin{algorithm}[t!]
	\small
	\DontPrintSemicolon
	\KwIn{Stratified Datalog program~$P = P_1\cup \cdots \cup P_n$, constraint~$\varphi$ over $P_n$}
	\KwOut{An input $I$ such that $\sem{P}_I \models \varphi$ or $\bot$}
	\Begin{
		$\badinputs_1 \gets \emptyset, \ldots, \badinputs_n \gets \emptyset$; $I_1 \gets \bot, \ldots, I_n \gets \bot$; $i\gets n$\;\label{line:bad-empty} \label{line:inputs-bot}
		\While{$i > 0$}{
			\If{$|\badinputs_i| > \textit{bound}_\badinputs$}{ \label{line:check-backtrack}				
				$\badinputs_i \gets \emptyset$; $\badinputs_{i+1}\gets \badinputs_{i+1} \cup \{I_{i+1}\}$\; 
				$i \gets i + 1$; {\hspace{10pt}\footnotesize\ttfamily // backtrack to higher stratum}\label{line:continue}\;\Continue 
			}
			
			$\psi_\badinputs \gets \bigwedge\limits_{I'\in \badinputs_i} \big( \neg \bigwedge\limits_{p\in \edb{P_i}} \encodepred{I'}{p} \big) $\; \label{line:avoid-bad}
			\If{$i=n$} {
				$\psi_i \gets \varphi$ \; \label{line:constraint-last-stratum}
			}
			\Else{
				$\psi_i \gets \bigwedge\limits_{p\in \edb{P_i}\cup \idb{P_i}} \encodepred{I_{i+1} \cup \cdots \cup I_n}{p}$\; \label{line:output-constraint}
			}
			$I_i = \synthbox(P_i, \psi_i \wedge \psi_\badinputs)$\; \label{line:gen-input}
			\If{$I_i \neq \bot$}{
				$i \gets i - 1$\;	\label{line:go-to-lower-stratum}						
			}
			\Else{
				\If{$i < n$}{
				$\badinputs_i \gets \emptyset$; $\badinputs_{i+1}\gets \badinputs_{i+1} \cup \{I_{i+1}\}$\; 
				$i \gets i + 1$ {\hspace{10pt}\footnotesize\ttfamily // backtrack to higher stratum} \label{line:go-to-higher-stratum}					
				} \Else{
				\Return{$\bot$} \label{line:return-unsat}
			}
			
		}
	}
	\Return{$I = \{ p(\overline{t})\in I_1\cup \cdots \cup I_n\mid p\in \edb{P} \}$}
}
\caption{Input synthesis algorithm~\synth for stratified Datalog}
\label{alg:stratified-datalog-synthesis}
\end{algorithm}
\setlength{\textfloatsep}{10pt}
Our iterative input synthesis algorithm for stratified Datalog, called \synth, is given in Algorithm~\ref{alg:stratified-datalog-synthesis}.
We assume that the fixed point constraint $\varphi$ is defined over predicates that appear in the highest stratum $P_n$;
this is without any loss of generality, as any constraint can be expressed using Datalog rules in the highest stratum, using a standard reduction to query satisfiability; cf.~\cite{Halevy:2001:SAD:502102.502104}.
Starting with the highest stratum~$P_n$, \synth generates an input~$I_n$ for $P_n$ such that $\sem{P_n}_{I_n}\models \varphi$. Then, it iteratively synthesizes an input for the lower strata $P_{n-1}, \ldots, P_1$ using the algorithm \synthbox. 
Finally, to construct an input for~$P$, the algorithm combines the inputs synthesized for all strata and returns this.

Recall that the fixed point of a stratum $P_i$ is given as input to the higher strata~$P_{i+1}, $ $\ldots, P_n$.
A key step when synthesizing an input~$I_i$ for~$P_i$ is thus to ensure that the {\em idb} predicates derived by $P_i$ are identical to the {\em edb} predicates synthesized for the inputs~$I_{i+1}, \ldots, I_n$ of the higher strata. 
Formally, let $$\Delta_i = \idb{P_i}\cap \edb{P_{i+1} \cup \cdots \cup P_n}$$
We must ensure that
$\{ p(\overline{t}) \in \sem{P_i}_{I_i} \mid p\in \Delta_i\} = \{ p(\overline{t}) \in I_{i+1} \cup \cdots \cup I_n\mid p\in \Delta_i\}$.

\para{Key Steps}
The algorithm first partitions $P$ into strata $P_1, \ldots P_n$. The strata can be computed using the predicates' dependency graph;  see~\cite[Chapter 15.2]{Abiteboul:1995:FDL:551350}.
For each stratum~$P_i$, it maintains a set of inputs~$\badinputs_i$, which contains inputs for~$P_i$ for which the algorithm failed to synthesize inputs for the lower strata $P_1, \ldots, P_{i-1}$. We call the sets $\badinputs_i$ {\em failed} inputs.
All $\badinputs_i$ are initially empty. 

In each iteration of the while loop, the algorithm attempts to generate an input~$I_i$ for stratum~$P_i$. 
At line~\ref{line:check-backtrack}, the algorithm checks whether $\badinputs_i$ has exceeded a pre-defined bound~$\textit{bound}_\badinputs$.
If the bound is exceeded, it adds $I_{i+1}$ to the failed inputs~$\badinputs_{i+1}$, re-initializes $\badinputs_i$ to the empty set, and backtracks to a higher stratum by incrementing $i$.
This avoids exhaustively searching through all inputs to find an input compatible with those synthesized for the higher strata.

At line~\ref{line:avoid-bad}, the algorithm uses the helper function $\encodepred{I'}{p}$. This function returns the constraint
$\forall \overline{X}.\ \big(\bigvee_{p(\overline{t})\in I'} \overline{X} = \overline{t}\big) \Leftrightarrow p(\overline{X})$, which is satisfied by an interpretation~$I$ iff $I$ contains identical $p(\overline{t})$ predicates as those in~$I'$. That is, if $I\models \encodepred{I'}{p}$ then for any $p(\overline{t})$ we have $p(\overline{t}) \in I$ iff $p(\overline{t})\in I'$. 
Therefore, the constraint $\psi_\badinputs$ constructed at line~\ref{line:avoid-bad} is satisfied by an input~$I_i$ iff $I_i\not\in \badinputs_i$, which avoids synthesizing inputs from the set of failed inputs.

The constraint $\psi_i$ in the algorithm constrains the fixed point of~$P_i$. For the highest stratum $P_n$, $\psi_i$ is set to the constraint $\varphi$ given as input to the algorithm. For the remaining strata~$P_i$, $\psi_i$ is satisfied iff the fixed point of $P_i$ is compatible with the synthesized inputs for the higher strata $P_{i+1}, \ldots, P_n$. In addition to constraining $P_i$'s \textit{idb} predicates, we also constraint the input \textit{edb} predicates. This is necessary to eagerly constrain the inputs.

At line~\ref{line:gen-input}, the algorithm invokes \synthbox to generate an input~$I_i$ such that $\sem{P_i}_{I_i} \models \varphi_i \wedge \psi_\badinputs$. The algorithm proceeds to the lower stratum if such an input is found ($I\neq \bot$); otherwise, if $i < n$  the algorithm backtracks to the higher stratum by increasing~$i$ and updating the sets $\badinputs_{i+1}$, and if $i = n$ if returns~$\bot$.

Finally, the while-loop terminates when the inputs of all strata have been generated. The algorithm constructs and returns the input~$I$ for $P$.

\begin{theorem}
	Let $P$ be a stratified Datalog program with strata $P_1, \ldots, P_n$, and $\varphi$ a constraint over predicates in $P_n$. If
	$\synth(P, \varphi) = I$ then $\sem{P}_I\models \varphi$. \footnote{\ifarxiv
	The theorem's proof can be found in Appendix~\ref{sec:proofs}.
\else
	The theorem's proof can be found in the technical report~\cite{synet-arxiv}.
\fi}
\end{theorem}

%We illustrate how \tool{} eagerly constraint the inputs with an example. Consider the Datalog program with strata $P_1$, $P_2$, and $P_3$, depicted in Figure~\ref{fig:eager-example}. Incoming and outgoing edges of a stratum~$P_i$ indicate the \textit{edb} predicates and, respectively, the \textit{idb} predicates of that stratum. For example, the stratum~$P_3$ takes as input predicates $q(\overline{t})$ and $r(\overline{t})$ and derives the predicate $s(\overline{t})$. The algorithm first synthesizes the input~$I_3$,  for the stratum $P_3$, which determines the predicates $q(\overline{t})$ and $r(\overline{t})$ that $P_1\cup P_2$ must output. When generating the constraint~$\psi_2$ (line~\ref{line:output-constraint} of algorithm~\synth), in addition to constraining the output predicate $r(\overline{t})$,  $\psi_2$ constrains the $P_2$'s input~\textit{edb} predicate $q(\overline{t})$. Suppose we do not constraint the \textit{edb} predicate $q(\overline{t})$ when generating~$I_2$. Then, the algorithm may end up asking~$P_1$ to output conflicting  (i.e., unsatisfiable) constraints for the predicate $q(\overline{t})$ whenever the $q(\overline{t})$ predicates synthesized for $I_1$ and $I_2$ do not match.
 \section{Implementation and Evaluation} \label{sec:evaluation}
In this section we first describe \tool{}, and end-to-end implementation of our input synthesis algorithm applied to the network-wide synthesis problem. We then turn to our evaluation of \tool{} on practical topologies and requirements.

\subsection{Implementation}\label{sec:implementation}

\tool{} is implemented in Python and automatically encodes stratified Datalog programs specified in the LogicBlox language~\cite{logicblox} into SMT constraints specified in the SMT-LIB v$2$ format~\cite{Barrett10c}.
It uses the Python API of Z3~\cite{DeMoura:2008:ZES:1792734.1792766} to check whether the generated SMT constraints are satisfiable and to obtain a model.

\tool{} supports routers that run both, OSPF and BGP protocols, and that can be configured with static routes.
\tool{} uses natural splitting for protocols: external routes are handled by BGP, while internal routes are handled by IGP protocols (OSPF and static, where static routes are preferred over OSPF).
We have partitioned the Datalog rules that capture these protocols and their dependencies into $8$ strata.
\tool{} relies on additional SMT constraints to ensure the well-formedness of the OSPF, BGP, and static route configurations output by our synthesizer.
For most topologies and requirements, the Datalog program reaches a fixed point within~$20$ iterations, and so we fixed the unroll and backtracking bounds~($\textit{bound}_k$ and $\textit{bound}_\badinputs$) to $20$.

\tool{} is vendor agnostic with respect to the synthesized configurations.
A simple script can be used to convert the output of \tool{} into any vendor specific configuration format and then deploy them in production routers. Indeed, to test the correctness of \tool{}, we implemented a small script to convert the input synthesized by \tool{} to Cisco router configurations.

\tool{} supports two key optimizations that improve its performance.
The first optimization is {\em partial evaluation}: \tool{} partially-evaluates Datalog rules with predicates whose truth values are known apriori. For example, all \pred{SetLink} predicates are known and can be eliminated. This reduces the number of variables in the rules and, in turn, in the generated SMT constraints.
The second optimization is {\em network-specific constraints}: we have configured \tool{} with generic constraints, which are true for all forwarding states, and with protocol-specific constraints, i.e. constraints that hold for any input to a particular protocol.
An example constraint is: {\em ``No packet is forwarded out of the router if the destination network is directly connected to the router''}.
These constraints are not specific to particular requirements or topology. They are thus defined one time and can be used to synthesize configurations for any requirements and networks.

\subsection{Experiments}\label{sec:experiments}

To investigate \tool's performance and scalability, we experimented with different:  \emph{(i)} topologies, \emph{(ii)}  requirements; and \emph{(iii)} protocol combinations.
Further to test correctness, we ran all synthesized configurations on an emulated environment of Cisco routers~\cite{gns3} and we verified that the forwarding paths computed match the requirements for each experiment.

\begin{wrapfigure}{R}{0.4\textwidth}
	\centering
	\includegraphics[width=0.4\textwidth]{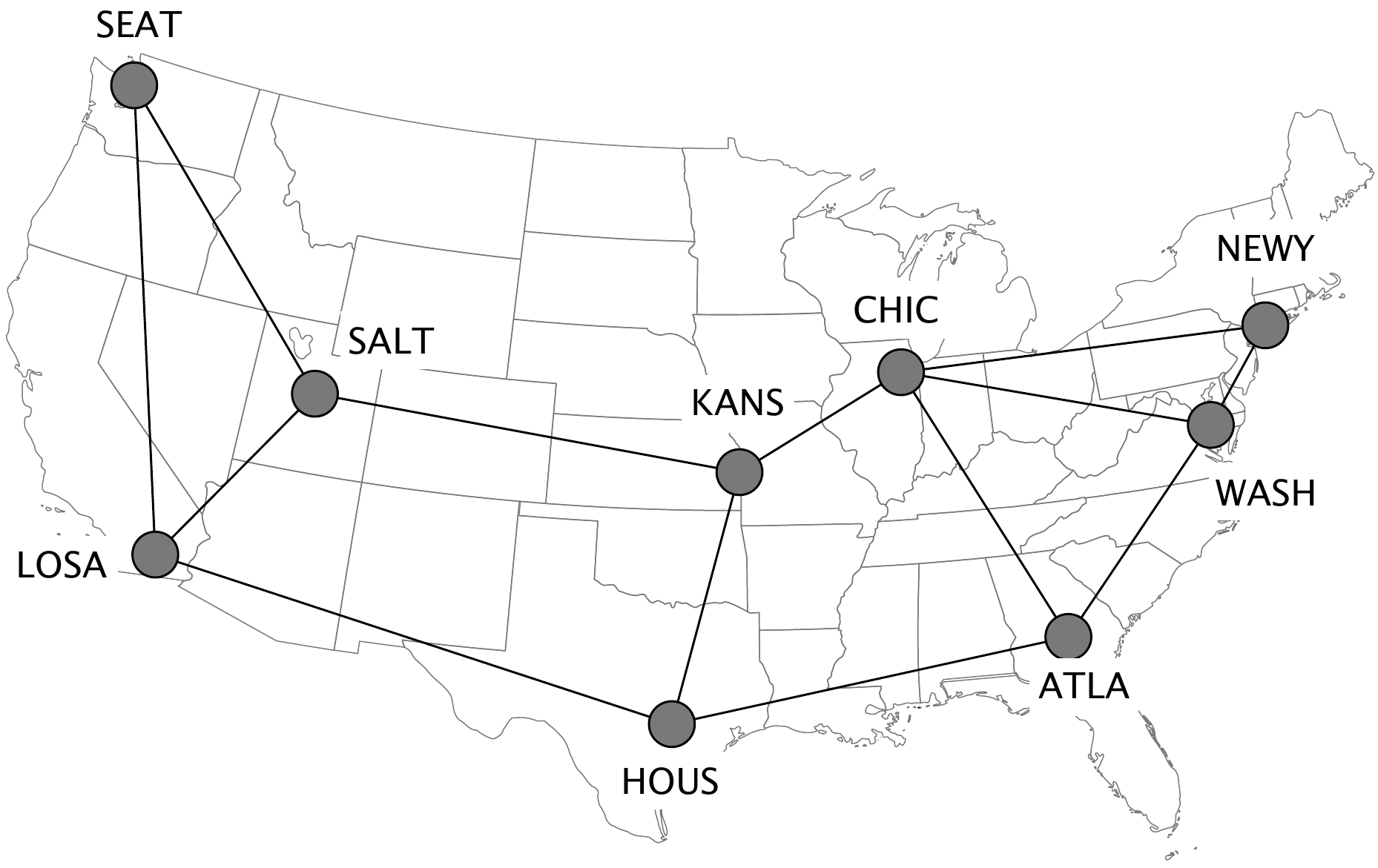}
	\caption{Internet$2$ topology}
	\label{fig:internet2}
\end{wrapfigure}

\para{Network Topologies} We used network topologies that have between $4$ and $64$ routers.
The $4$-router network is our overview example where we considered the same requirements as those described in Section~\ref{sec:overview}.
The $9$-router network is Internet$2$ (see Figure~\ref{fig:internet2}), a US-based network that connects several major universities and research institutes.
The remaining networks are $n\times n$ grids.

\para{Routing Requirements}
For each router and each traffic class, we generate a routing requirement that defines where the packets for that traffic class must be forwarded to. 
We consider $1$, $5$, and $10$ traffic classes.
For a topology with $n$ routers and $m$ traffic classes, we thus generate $n\times m$ requirements.

For topologies with multiple traffic classes, we add one external network announced by two randomly selected routers.
We add requirements to enforce that all packets destined to the external networks are forwarded to one of the two routers.
This models a scenario where the operator is planning maintenance downtime for one of the two routers.
Further, to show that \tool{} synthesizes configurations with partially defined input and protocol dependencies, we assume the local BGP preferences are fixed by the network operator and thus \tool{} has to synthesize correct OSPF costs to meet the BGP requirements.

\para{Protocols}
We consider three different combinations of protocols: \emph{(i)} static routes; \emph{(ii)} OSPF and static routes; and \emph{(iii)} OSPF, BGP, and static routes.
The protocol combinations \emph{(i)} and \emph{(ii)} ignore requirements for external networks since only BGP computes routes for them.

\setlength{\tabcolsep}{5pt}
\begin{table}[!t]
\centering
\resizebox{\textwidth}{!}{\begin{tabular}{@{}cllrrrrrrrr@{}}
\toprule
                   &  &   & \multicolumn{2}{c}{1 Traffic Class} &  &  \multicolumn{2}{c}{5 Traffic Classes}   &  & \multicolumn{2}{c}{10 Traffic Classes}  \\  \cmidrule(l){2-11}
Protocol          & \# Routers  &  & Avg  & Std &  &  Avg  & Std &  & Avg & Std  \\  \cmidrule(l){1-11}
Static  & $9$                   & &              $1.3$s &    ($0.5$) & &             $2.0$s &    ($0.1$) & &             $2.8$s & ($0.4$) \\
           & $9$ (Internet2) & &               $1.3$s &   ($0.5$) & &              $2.0$s &   ($0.0$) & &              $4.0$s & ($0.8$) \\
           & $16$                 & &               $5.9$s &   ($0.3$) & &              $7.8$s &   ($0.4$) & &            $11.2$s & ($0.4$) \\
           & $25$                 & &             $32.0$s &   ($0.6$) & &            $37.0$s &   ($0.6$) & &            $46.1$s & ($0.9$) \\
           & $36$                 & &    $2$m$49.7$s  &  ($3.0$) & &     $3$m$1.5$s &   ($4.5$) & &       $3m27.0$s&  ($4.4$) \\
           & $49$                 & &  $12$m$29.2$s &   ($7.0$) & & $13$m$02.3$s & ($10.6$) & & $14$m$10.7$s & ($15.0$) \\
           & $64$                 & &  $46$m$36.2$s & ($49.0$) & & $47$m$23.8$s & ($27.2$) & & $49$m$22.2$s & ($39.3$) \\
\cmidrule(l){2-11}
OSPF+Static  & $9$                   & &                       $9.4$s &  (0$.5$)     & &                     $19.8$s &     ($0.4$)  & &                      $39.9$s & ($0.5$) \\
                       & $9$ (Internet2) & &                       $9.0$s &  ($1.4$)     & &                     $21.3$s &      ($1.2$)  & &                     $49.3$s &  ($0.5$) \\
                       & $16$                 & &                     $43.5$s &  ($0.7$)     & &            $1$m$19.8$s &      ($0.6$)  & &              $4$m$5.8$s &  ($1.6$) \\
                       & $25$                 & &            $2$m$55.2$s &  ($6.1$)     & &              $7$m$3.8$s &      ($9.9$)  & &          $15$m$56.4$s &  ($38.1$) \\
                       & $36$                 & &          $10$m$00.5$s &  ($9.5$)     & &          $23$m$58.9$s &    ($22.5$)  & &   $1$h$11$m$38.2$s &  ($127.5$) \\
                       & $49$                 & &          $24$m$11.6$s &  ($43.5$)   & &  $1$h$30$m$00.3$s &    ($89.6$)  & &   $5$h$22$m$55.8$s &  ($421.2$) \\
                       & $64$                 & &  $2$h$22$m$13.2$s &  ($209.9$) & &  $5$h$42$m$58.9$s &  ($619.4$)  & & $21$h$13$m$16.0$s &  ($1986.7$) \\
                       \cmidrule(l){2-11} 
BGP+OSPF+Static  & $9$                  & &                     $15.3$s &  ($0.5$)      & &                   $27.7$s  &  ($0.5$) & &  $1$m$0.5$s  & ($2.6$) \\
                                & $9$ (Internet2) & &                     $13.3$s &  ($0.9$)      & &                   $22.7$s &  ($0.9$) & &  $1$m$19.7$s &  ($0.5$) \\
                                & $16$                 & &                     $56.0$s &  ($1.6$)      & &          $2$m$24.7$s &  ($0.9$) & &  $8$m$29.0$s &  ($10.7$) \\
                                & $25$                 & &            $3$m$56.3$s &  ($3.1$)      & &          $8$m$46.3$s &   ($5.3$) & & $40$m$09.3$s &  ($99.2$) \\
                                & $36$                 & &          $14$m$14.0$s &  ($15.0$)    & &        $43$m$38.0$s &  ($5.7$) & & $2$h$35$m$11.7$s &  ($197.7$) \\
                                & $49$                 & &   $1$h$23$m$20.7$s &  ($211.1$) & &$2$h$15$m$18.0$s &  ($12.8$) & & timeout ($>$ $24$h) \\
                                & $64$                 & &   $1$h$46$m$35.0$s &  ($165.8$) && $7$h$24$m$51.3$s & ($519.2$) & & timeout ($>$ $24$h)\\
\bottomrule
 
\end{tabular}}
\vspace{4pt}
\caption{\tool{}'s synthesis times (averaged over 10 runs) for different number of routers, protocol combinations, and traffic classes in the requirements.
}
\label{tbl:results}
\end{table}

\para{Experimental Setup}
We run \tool{} on a machine with $128$GB of RAM and a modern $12$-core dual-processors running at $2.3$GHz.

\para{Results}
The synthesis times for the different networks and protocol combinations are
shown in Table~\ref{tbl:results} (averaged over 10 runs). \tool{} synthesizes
the overview example's configuration described in Section~\ref{sec:overview} in
$10$ seconds. For the largest network ($64$ routers) and number of traffic
classes ($10$ classes), \tool{} synthesizes a configuration for static routes
(protocol combination \emph{(i)}) in less than $1$h, and for the combination of
static routes and OSPF, \tool{} takes less than $22$h. 
When using both OSPF and BGP protocols along with static routes, for all network topologies \tool{} synthesizes configurations for $1$ and $5$ traffic classes within $8$h; for $10$ traffic classes, \tool{} times out after $24$h for the largest topologies with $49$ and $64$ routers.

\para{Interpretation} Our results show that \tool{} scales to real-world networks. Indeed, a longitudinal analysis of more than $260$ production networks~\cite{topozoo} revealed that $56\%$ of them have less than $32$ routers.
\tool{} would synthesize configurations for such networks within one hour. \tool{} also already supports a reasonable amount of traffic classes. According to a study on real-world
enterprise and WAN networks~\cite{miningpolicies}, even large networks with
$100${,}$000$s of IP prefixes in their forwarding tables usually see less than $15$ traffic classes in total.

While \tool{} can take more than 24 hours to synthesize a configuration for the
largest networks (with all protocols activated and 10 traffic classes), we believe that this time can
be reduced through divide-and-conquer. Real networks tend to be
hierarchically organized around few regions (to ensure the scalability of the protocols~\cite{Doyle:2005:RTV:1076956}) whose
configurations can be synthesized independently. We plan to explore the
synthesis of such hierarchical configurations in future work.

%in \emph{less than 20 minutes}. 
 \newcommand{\muze}{$\mu Z$\xspace}

\section{Related Work}\label{sec:related}

\para{Analysis of Datalog Programs}
Datalog has been successfully used to declaratively specify variety of static analyzers~\cite{SmaragdakisB10, ZhangMGNY14}. It has been also used to verify network-wide configurations for protocols such as OSPF and BGP~\cite{Fogel:2015:GAN:2789770.2789803}.
Recent work~\cite{Madsen:2016} has extended Datalog to operate with richer classes of lattice structures.
Further, the \muze tool~\cite{Hoder:2011:EEF:2032305.2032341} extends the Z3 SMT solver with support for fixed points.
The focus of all these works is on computing the fixed point of a program~$P$ for a given input~$I$ and then checking a property~$\varphi$ on the fixed point. That is, they check whether $\sem{P}_I\models \varphi$.
All of these works assume that the input is provided a priori.
In contrast, our procedure discovers an input that produces a fixed point satisfying a given (user-provided) property on the fixed point.

The algorithm presented in~\cite{ZhangMGNY14} can be used to check whether certain tuples are not derived for a given set of inputs.
Given a Datalog program~$P$ (without negation in the literals), a set~$Q$ of tuples, and a set~$\cal I$ of inputs, the algorithm computes the set $Q\setminus \bigcap \{ \sem{P}_I\mid I\in {\cal I} \}$.
This algorithm cannot address our problem because it does not support stratified Datalog programs, which are not monotone.
While their encoding can be used to synthesize inputs for each stratum of a stratified Datalog program, it supports only negative properties, which require that certain tuples are not derived.
Our approach is thus more general than \cite{ZhangMGNY14} and can be used in their application domain.

The FORMULA system~\cite{Jackson:2006:TFF:1176887.1176896,Jackson:2008:MGH:1424493.1424495} can synthesize inputs for non-recursive Dataog programs, as it supports non-recursive Horn clauses with stratified negation (even though~\cite{Jackson:2010:CPP:1879021.1879027} which uses FORMULA shows examples of recursive Horn clauses w/o negation). Handling recursion with stratified negation is nontrivial as bounded unrolling is unsound if applied to all strata together. Note that virtually all network specifications require recursive rules, which our system supports. 

\para{Symbolic Analysis and Synthesis}
Our algorithm is similar in spirit to symbolic (or concolic) execution, which is used to automatically generate inputs for programs that violate a given assertion (e.g. division by zero); see~\cite{Cadar:2013:SES:2408776.2408795,Kroening2014,Clarke2004} for an overview.
These approaches unroll loops up to a bound and find inputs by calling an SMT solver on the symbolic path.
While we also find inputs for a symbolic formula, the entire setting, techniques and algorithms, are all different from the standard symbolic execution setting.

Counter-example guided synthesis approaches 
are also related~\cite{Solar-Lezama:2006:CSF:1168857.1168907}.
Typically, the goal of synthesis is to discover a program, while in our case the program is given and we synthesize an input for it.
There is a connection, however, as a program can be represented as a vector of bits.
Most such approaches have a single counter-example generator (i.e., the oracle), while we use a sequence of oracles.
It would be interesting to investigate domains where such layered oracle counter-example generation can benefit and improve the efficiency of synthesis.

\para{Network configuration synthesis}
Propane~\cite{propane2016} and Genesis~\cite{Subramanian:2017:GSF:3009837.3009845} also produce network-wide configurations out of routing requirements. Unlike our approach, however, Propane only supports BGP and Genesis only supports static routes.
In contrast to our system, Propane and Genesis support failure-resilience requirements.
While we could directly capture such requirements by quantifying over links, this would make synthesis more expensive.
A more efficient way to handle such requirements would be to synthesize a failure-resilient forwarding plane using a system like Genesis~\cite{Subramanian:2017:GSF:3009837.3009845}, and to then feed this as input to our synthesizer to get a network-wide configuration.
In contrast to these approaches, our system is more general:
one can directly extended it with additional routing protocols, by specifying them in stratified Datalog, and synthesize configurations for \emph{any combination} of routing protocols.

ConfigAssure~\cite{config_assure} is a general system that takes as input requirements in first-order constraints and outputs a configuration conforming to the requirements.
The fixed point computation performed by routing protocols cannot be captured using the formalism used in ConfigAssure.
Therefore, ConfigAssure cannot be used to specify networks and, in turn, to synthesize protocol configurations for networks.  \section{Conclusion}

We formulated the network-wide configuration synthesis problem as a problem of finding inputs of a Datalog program, and presented a new input synthesis algorithm to solve this challenge. Our algorithm is based on decomposing the Datalog rules into strata and iteratively synthesizing inputs for the individual strata using off-the-shelf SMT solvers.
We implemented our approach in a system called \tool~and showed that it scales to realistic network size using any combination of OSPF, BGP and static routes. %Network operators can now express their global routing requirements and use \tool~to automatically obtain a network-wide configuration which ensures that routers compute a compliant forwarding state.
 
{\footnotesize
	\bibliographystyle{unsrtnat}
		\bibliography{bib}
					}

\newpage
\appendix

\ifarxiv
\section{Formalizing OSPF}

\label{sec:ospf-specification}

\begin{figure}[t]
		\footnotesize
	\begin{verbatim}	
	BestOSPFRoute(TC, Router, NextHop) :- minCost(TC, Router, Cost),
	OSPFRoute(TC, Router, NextHop, Cost)
	minCost(TC, Router, min<Cost>) :- OSPFRoute(TC, Router, NextHop, Cost)
	OSPFRoute(TC, Router, NextHop, Cost) :- SetNetwork(_, Net),
	SetOSPFEdgeCost(Router, NextHop, Cost)
	OSPFRoute(TC, Router, NextHop, Cost) :- Cost = Cost1 + Cost2
	SetOSPFEdgeCost(Router, NextHop, Cost1),
	OSPFRoute(TC, NextHop, R', Cost2)
	\end{verbatim}
	\vspace{-14pt}
	\caption{Declarative specification of the OSPF protocol}
	\vspace{10pt}
	\label{fig:ospf}
\end{figure}

In Figure~\ref{fig:ospf}, we show (a subset of) of our OSPF formalization in stratified Datalog.
The predicate \pred{BestOSPFRoute(TC, Router, NextHop, Cost)} represents the best OSPF route selected by the router \pred{Router} for the network \pred{Net} to be the next hop \pred{NextHop} associated with the minimum cost \pred{Cost}. This behavior is formalized with the first rule in Figure~\ref{fig:ospf}.
The second rule derives the minimum cost OSPF route for each router and each destination network by aggregating over all possible OSPF routes.
Finally, the last two rules concisely implement the shortest-path computation performed by the routers running OSPF. The predicate \pred{SetOSPFEdgeCost(R1, R2, Cost)} represents that the routers \pred{R1} and \pred{R2} are neighbors connected by a link with cost \pred{Cost}, and the predicate
\pred{SetNetwork} for any value that represents a network.
The third rule thus formalizes that \pred{R1} can forward packets to \pred{R2}, for any network~\pred{Net}, with this cost.
The last rule transitively computes multi-hop routing paths by summing up the costs associated along all OSPF routes. \newcommand{\bound}{\textit{bound}}
\newcommand{\inp}[2]{{\textit{inp}(#1,#2)}}

\section{Proofs}
\label{sec:proofs}
We start with a couple of preliminary definitions.

Given a program~$P$ and an interpretation $J$, we denote by
$\inp{P}{J}$ the set of all ground atoms contained in $J$ that are constructed with \textit{edb} predicate symbols of the program~$P$.
Formally, $\inp{J}{P} = \{p(\overline{t})\mid p\in \textit{edb}(P)\}$.

\subsection{Semi-positive Algorithm}

First, we remark that any semi-positive Datalog program can be stratified into a single partition~$P$.
The model~$\sem{P}_I$ of $P$ for a given input~$I$ is given by the least fixed point of the consequence operator $T_P$ that contains $I$.
The fixed point $\sem{P}_I$ can be iteratively computed as $T_{P,I}^\infty$ where 
$T_{P,I}^0 = I$ and
$T_{P,I}^{i+1} = T_{P}(T_{P,I}^i) \cup T_{P,I}^i$.
Note that $T_{P, I}^i \subseteq T_{P,I}^j$ for any $i \leq j$.

\para{Negative Constraints}
We first show that any interpretation $J$ that satisfies the constraint $\translate{P}{k}$ is an over-approximation of the ground atoms $p(\overline{t})$ derived by program $P$ for the input $\inp{P}{J}$.

\begin{lemma}
	Let~$P$ be a semi-positive Datalog program. For any~$k\geq 0$ and any interpretation~$J$ such that $J\models \translate{P}{k}$, we have $\sem{P}_\inp{P}{J} \subseteq J$.
	\label{lemma-neg}
\end{lemma}

\begin{proof}
	By induction on the iterative computation of $\sem{P}_\inp{P}{J}$, we show that for any $i\geq 0$ we have $T_{P,\inp{P}{J}}^i\subseteq J$.
	
	{\bf Base Case:} For the base case, we have $i = 0$. Then, $T_{P, \inp{P}{J}}^0 = \inp{P}{J}$. 	
	Since $\inp{P}{J} = \{p(\overline{t})\in J\mid p\in \edb{P}\}$, it is immediate that $\inp{P}{J} \subseteq J$, and thus $T_{P, \inp{P}{J}}^0 \subseteq J$.
	
	{\bf Inductive Step:} For our inductive step, assume that $T_{P,\inp{P}{J}}^j\subseteq J$ holds for $0\leq j \leq i$, for some $i\geq 0$. We show that $T_{P,\inp{P}{J}}^{i+1}\subseteq J$. 
		
	By definition, we have $T_{P,\inp{P}{J}}^{i+1} = T_{P}(T_{P,\inp{P}{J}}^i) \cup T_{P,\inp{P}{J}}^i$.
	By induction, we know that $T_{P,\inp{P}{J}}^j\subseteq J$. It remains to prove that $T_{P}(T_{P,\inp{P}{J}}^i)\subseteq J$.
	Suppose $p(\overline{t})\in T_{P}(T_{P,\inp{P}{J}}^i)$. We need to show that $p(\overline{t})\in J$.
	Since $p(\overline{t})\in T_{P}(T_{P,\inp{P}{J}}^i)$, we know that there is a rule $p(\overline{X})\leftarrow l_1, \ldots, l_n$ in $P$ such that 
	for some substitution $\sigma$ we have $\sigma(p(\overline{X})) = p(\overline{t})$ and for all $l_i$ we have $T_{P,\inp{P}{J}}^i \vdash \sigma(l_i)$.
	We can conclude that $T_{P,\inp{P}{J}}^i \models \exists \overline{Y}. l_1\wedge \cdots \wedge l_n$. By induction hypothesis, we have $T_{P,\inp{P}{J}}^j\subseteq J$. Since $P$ is semi-positive, we know that all negative literals in $l_1, \ldots, l_n$ are constructed using \textit{edb} predicates. Moreover, both $T_{P,\inp{P}{J}}$ and $J$ contain the same set of \textit{edb} literals, and we can thus conclude that $J\models \exists \overline{Y}. l_1\wedge \cdots \wedge l_n$.
	By definition of $\translate{P}{k}$, we know that $\translate{P}{k}$ contains the constraint 
	$\forall \overline{X}.\ ( (\exists \overline{Y}.l_1\wedge \cdots \wedge l_n) \Rightarrow p(\overline{X}) )$. Since $J\models \translate{P}{k}$, we get that $J\models p(\overline{t})$. Therefore, $p(\overline{t})\in J$.
	\qed
\end{proof}

We can now prove that $\synthbox$ is sound for negative constraints.

\begin{lemma}
	Let $P$ be a semi-positive Datalog program and $\neg p(\overline{t})$ a negative constraint. If $\synthbox(P, \neg p(\overline{t})) = I$, then $\sem{P}_I\models \neg p(\overline{t})$.
	\label{lemma-neg-2}
\end{lemma}

\begin{proof}
	Suppose $\synthbox$ returns an input $I$ for some $k\in [1.. \bound_k]$. The input $I$ is derived from an interpretation $J$ such that $J\models \translate{P}{k}\wedge \neg p(\overline{t})$
	and $\inp{P}{J} = I$.
		From $J\models \neg p(\overline{t})$, we get $p(\overline{t})\not\in J$. Furthermore, from $J\models \translate{P}{k}$, by Lemma~\ref{lemma-neg}, we get $\sem{P}_I\subseteq J$. We conclude that $p(\overline{t})\not\in \sem{P}_I$ and thus $\sem{P}_I\models \neg p(\overline{t})$.\qed
\end{proof}

\para{Positive Constraints}
We now prove that any interpretation $J$ that satisfies the constraint $\translate{P}{k}$ contains a ground atom $p_k(\overline{t})$ then the ground atom $p(\overline{t})$ is derived by $P$ for input $\inp{P}{J}$.

\begin{lemma}
	Let $P$ be a semi-positive Datalog program. For any $k\geq 1$ and any interpretation $J$ such that $J\models \translate{P}{k}$, if $p_k(\overline{t})\in J$ then $p(\overline{t})\in \sem{P}_\inp{P}{J}$.
	\label{lemma-pos}
\end{lemma}

\begin{proof}
	By induction on the iterative computation of $\sem{P}_I$, we show that $p_i(\overline{t})\in J$ implies that $p(\overline{t}) \in T_{P,\inp{P}{J}}^i$, for any $i\geq 1$. Since $T_{P, \inp{P}{J}}^i\subseteq \sem{P}_\inp{P}{J}$ for any $i$, this also implies that 
	$p(\overline{t}) \in \sem{P}_\inp{P}{J}$.	
	
	{\bf Base Case:} For the base case, we have $i = 1$. Assume $p_1(\overline{t})\in J$. By definition of $\translate{P}{1}$, the constraint \[\forall \overline{X}. \big( p_1(\overline{X})\Leftrightarrow  ( \bigvee\limits_{p(\overline{X})\leftarrow \overline{l} \in P} \exists \overline{Y} .\ \tau(\overline{l}, 0))
	\big),\] where $\overline{Y} = \getvars{\overline{l}}\setminus \overline{X}$, is conjoined to the constraint $\translate{P}{1}$. Since $J\models p_1(\overline{t})$, we conclude that there is a rule $p(\overline{X})\leftarrow l_1, \ldots, l_n$ in $P$ such that
	for some substitution $\sigma$ we have $\sigma(p(\overline{X})) = p(\overline{t})$ and $J\models \sigma(\tau(l_i, 0))$ for $1\leq i \leq n$. 
		By definition of $\tau$, all literals $l_i$ must be constructed using \textit{edb} predicates (since $\tau(l_i, 0)$ maps any idb literal $l_i$ to $\false$ and $J'\not\models \false$ for any $J'$). Note that for \textit{edb} literals we have $\tau(l_i, 0) = l_i$.
	Since $J$ and $\inp{P}{J}$ contain the same set of \textit{edb} ground atoms, we get $\inp{P}{J} \vdash \sigma(l_i)$ for all $1\leq i \leq n$. By definition of $T_P$ and $T_{P,\inp{P}{J}}^1$, it is immediate that $p(\overline{t})\in T_{P,\inp{P}{J}}^1$.
	
	{\bf Inductive Step:} For our inductive step,  assume that $p_j(\overline{t})\in J$ implies that $p(\overline{t}) \in T_{P,\inp{P}{J}}^j$, for $1 \leq j \leq i$, for some $i\geq 1$. We show that $p_{i+1}(\overline{t})\in J$ implies that $p(\overline{t}) \in T_{P,\inp{P}{J}}^{i+1}$.
	
	Assume $p_{i+1}(\overline{t})\in J$.
	By definition of $\translate{P}{i+1}$, the constraint \[\forall \overline{X}. \big( p_{i+1}(\overline{X})\Leftrightarrow ( \bigvee\limits_{p(\overline{X})\leftarrow \overline{l} \in P} \exists \overline{Y}. \tau(\overline{l}, i))
	\big),\] where $\overline{Y} = \getvars{\overline{l}}\setminus \overline{X}$, is conjoined to the constraint $\translate{P}{i+1}$.
		Since $p_{i+1}(\overline{t})\in J$ and $J\models \translate{P}{i+1}$, we know there is a rule $p(\overline{X})\leftarrow l_1, \ldots, l_n$ in $P$ such that 
	for some substitution $\sigma$ we have $\sigma(p(\overline{X})) = p(\overline{t})$ and
		$J\models \sigma(\tau(l_1, i)) \wedge \cdots \wedge \sigma(\tau(l_n, i))$. For any \textit{edb} literal $l$ in the body of this rule, we have $\tau(l, i) = l$ and $\sigma(l) \in J$ iff $\sigma(l)\in T_{P,\inp{P}{J}}^i$, simply because $J$ and $T_{P,\inp{P}{J}}^i$ contain the same \textit{edb} ground atoms.
		Furthermore, for any positive \textit{idb} literal $\sigma(l) = q(\overline{t'})\in J$ in the body of this rule, 
	we have $\tau(q(\overline{t'}), i) = q_i(\overline{t'})$ and
		using our inductive hypothesis we get $q(\overline{t})\in T_{P,\inp{P}{J}}^i$.
		We conclude for all literals $l$ that appear in the body of this rule we have $T_{P,\inp{P}{J}}^i\vdash \sigma(l)$. By definition of $T_{P,\inp{P}{J}}^{i+1}$ and $T_P$ we conclude that $p(\overline{t})\in T_{P,\inp{P}{J}}^{i+1}$. 
\qed
\end{proof}

We can now prove that $\synthbox$ is sound for positive constraints.

\begin{lemma}
	Let $P$ be a semi-positive Datalog program and $p(\overline{t})$ a positive constraint. If $\synthbox(P, p(\overline{t})) = I$ then  $\sem{P}_I \models p(\overline{t})$.
	\label{lemma-pos-2}
\end{lemma}

\begin{proof}
	Suppose $\synthbox$ returns an input $I$ for some $k\in [1.. \bound_k]$. The input $I$ is derived from an interpretation $J$ such that $\inp{P}{J} = I$ and $J\models \translate{P}{k}\wedge p_k(\overline{t})$
	From $J\models p_k(\overline{t})$, we know that $p_k(\overline{t}) \in J$. From $J\models \translate{P}{k}$, by Lemma~\ref{lemma-pos}, we get $p(\overline{t})\in \sem{P}_\inp{P}{J}$. It is immediate that $\sem{P}_I\models p(\overline{t})$.
\end{proof}

We can now prove the correctness of $\synthbox$.

\smallskip
\noindent
{\bf Theorem 1.}\ Let $P$ be a semi-positive Datalog program and $\varphi$ a constraint. If $\synthbox(P, \varphi) = I$ then $\sem{P}_I\models \varphi$. 

\begin{proof}
	The algorithm $\synthbox$ transforms the constraint $\varphi$ into a constraint that uses conjunction and disjunction over positive and negative constraints.
		Since conjunction and disjunction and monotone,
	the proof of $\sem{P}_I\models \varphi$ follows from Lemma~\ref{lemma-neg-2} and Lemma~\ref{lemma-pos-2}.\qed
\end{proof}

\subsection{Stratified Algorithm}

We now prove the correctness of the stratified input synthesis algorithm~$\synth$, which uses the $\synthbox$ algorithm as a building block. Given an interpretation $I$ and 

\smallskip
\noindent
{\bf Theorem 2.}\ 
Let $P$ be a stratified Datalog program with strata $P_1, \ldots, P_n$, and $\varphi$ a constraint over predicates in $P_n$. If
$\synth(P, \varphi) = I$ then $\sem{P}_I\models \varphi$.

\begin{proof}
By induction on the computation of the inputs $I_n, I_{n-1}, \ldots, I_1$, we show that
$\sem{P_i \cup \cdots \cup P_n}_\inp{P_i \cup \cdots \cup P_n}{I_i\cup \cdots \cup I_n}  \models \varphi$ holds for $1\leq i \leq n$. Note that the case for $i=1$ proves the theorem.

{\bf Base Case:} For the base case, we have $i = n$. Then $I_n = \synthbox(P_n, \varphi)$. 
We have $\inp{P_n}{I_n} = I_n$, and by Theorem~1, we get $\sem{P_n}_{I_n}\models \varphi$.

{\bf Inductive Step:}
For our inductive step, assume that \\
$\sem{P_j \cup \cdots \cup P_n}_\inp{P_j \cup \cdots \cup P_n}{I_j\cup \cdots \cup I_n}  \models \varphi$ holds for $i\leq j \leq n$,
for some $1<i\leq n$.
We need to show that $\sem{P_{i-1}\cup \cdots \cup P_n}_\inp{P_{i-1}\cup \cdots \cup P_n}{I_{i-1}\cup \cdots \cup I_n} \models \varphi$.
Recall that according to the semantics of stratified Datalog, the model $\sem{P_{i-1}\cup \cdots \cup P_n}_\inp{P_{i-1}\cup \cdots \cup P_n}{I_{i-1}\cup \cdots \cup I_n}$ is computed by first computing 
$\sem{P_{i-1}}_{I_{i-1}}$ and then computing $\sem{P_i\cup \cdots \cup P_n}_I$ where $I$ contains all ground atoms in $\sem{P_{i-1}}_{I_{i-1}}$ together with ground atoms in $I' = \inp{P_{i-1}\cup \cdots \cup P_n}{I_i\cup \cdots \cup I_n}$. 
The only difference between $I$ and $\inp{P_i\cup \cdots \cup P_n}{I_i\cup \cdots \cup I_n}$ therefore is that \textit{edb} atoms of $P_i\cup \cdots \cup P_n$ that are contained in $I_i\cup \dots I_n$ and are constructed using \textit{idb} predicates of $P_{i-1}$ are now derived by the the program $P_{i-1}$ for the input $I_{i-1}$. 
The constraint $\varphi_{i-1}$ constructed at line~\ref{line:output-constraint} of Algorithm~\ref{alg:stratified-datalog-synthesis} ensures that these two sets of ground atoms are identical. 
We can thus conclude that
$\sem{P_{i-1}\cup P_i \cdots \cup P_n}_\inp{P_{i-1}\cup \cdots \cup P_n}{I_{i-1}\cup \cdots \cup I_n} \models \varphi$.
\qed
\end{proof} \fi

\end{document}